\documentclass[12pt,reqno]{amsart}
\usepackage[utf8]{inputenc}
\usepackage{amsmath,amsfonts,amsbsy,amsthm}
\usepackage{eucal}
\usepackage{dsfont}
\usepackage{comment}

\usepackage{xy}
\xyoption{matrix} \xyoption{arrow} \xyoption{arc} \xyoption{color}

\UseCrayolaColors
\usepackage{enumerate}
\usepackage[shortlabels]{enumitem}
\setlist{leftmargin=*}

\usepackage{tikz}
\usetikzlibrary{decorations.pathmorphing,shapes,arrows}

\newcounter{sarrow}

\numberwithin{equation}{section}

\makeatletter
\newtheoremstyle{corsivo}
   {\medskipamount}{\medskipamount}%
   {\itshape}{}%
   {\bfseries}{}%
   { }
   {\thmname{#1}\thmnumber{\@ifnotempty{#1}{ }\@upn{#2}}%
    \thmnote{ {\bfseries(#3)}}.}%
\makeatother

\theoremstyle{corsivo}
\newtheorem{thm}{Theorem}[section]

\newtheorem{prop}[thm]{Proposition}

\makeatletter
\newtheoremstyle{dritto}
   {\medskipamount}{\medskipamount}%
   {\rmfamily}{}%
   {\bfseries}{}%
   { }
   {\thmname{#1}\thmnumber{\@ifnotempty{#1}{ }\@upn{#2}}%
    \thmnote{ {\bfseries(#3)}}.}%
\makeatother

\theoremstyle{dritto}
\newtheorem{dfn}[thm]{Definition}
\newtheorem{rmk}[thm]{Remark}

\newcommand{\V}{P}

\newcommand{\eu}{\mathrm{e}}
\newcommand{\iu}{\mathrm{i}}   
\newcommand{\di}{\mathrm{d}}

\newcommand{\Hol}{\text{Hol}}

\newcommand{\N}{\mathbb{N}}
\newcommand{\Z}{\mathbb{Z}}

\newcommand{\R}{\mathbb{R}}
\newcommand{\C}{\mathbb{C}}

\newcommand{\inner}[2]{\left\langle #1\,|\, #2 \right\rangle}  

\newcommand{\norm}[1]{\left\| #1 \right\|}

\newcommand{\set}[1]{ \left\{  #1 \right\}} 

\DeclareMathOperator{\Tr}{Tr}         

\DeclareMathOperator{\Ran}{Ran}

\DeclareMathOperator{\dist}{dist}

\newcommand{\ie}{{\sl i.\,e.\ }}   

\newcommand{\virg}[1]{``#1''}

\newcommand{\half}{\mbox{\footnotesize $\frac{1}{2}$}}

\newcommand{\w}{\psi}  

\newcommand{\jp}[1]{\langle#1\rangle}

\newcommand{\Lattice}{\mathfrak{D}}

\newcommand{\x}{{\bf x}}
\newcommand{\z}{{\bf z}}

\newcommand{\y}{{\bf y}}

\usepackage{color}

\let\oldfootnote\footnote
\renewcommand{\footnote}[1]{\oldfootnote{\  #1}}

\title[Ultra-generalized Wannier Bases and topological transport]
{ Ultra-generalized Wannier bases: \\[2mm] are they relevant to topological transport? }

\author{Massimo Moscolari and Gianluca Panati}

\date{ \today. Final version for arXiv. Paper published in J. Math. Phys. 64, 071901 (2023). DOI: 10.1063/5.0137320
}

\usepackage{hyperref}

\begin{document}
	
\maketitle
\vspace{-1cm}

\begin{abstract}
We generalize Prodan's construction of radially localized generalized Wannier bases [E. Prodan, On the generalized Wannier functions. {\it J. Math. Phys.} {\bf 56}(11), 113511 (2015)] to gapped quantum systems without time-reversal symmetry, including in particular magnetic Schr\"odinger operators, and we prove some basic properties of such bases. We investigate whether this notion might be relevant to topological transport by considering the explicitly solvable case of the Landau operator.
\end{abstract}

\section{Introduction}

Wannier functions, and their generalizations, are nowadays a fundamental tool in solid-state physics \cite{WannierReview}.  
Whenever a basis of well-localized generalized Wannier functions exists, it allows computational methods whose cost scales only linearly with respect to the system size \cite{Goedecker}, as well as an intuitive understanding of polarization and orbital magnetization in solids \cite{CeresoliThonhauserVanderbiltResta2006}. 

A few years ago, it has been noticed that Wannier bases can be used to \virg{detect} topological phases of matter, in the sense that they allow to discriminate between ordinary and Chern insulators. In other words, Wannier bases are special orthonormal bases for the range of the Fermi projection of a gapped quantum system, that are able to distinguish whether the Chern number of the projection is vanishing or not. This follows from the so-called \emph{Localization Dichotomy},  initially stated and proved for gapped $\Z^d$-periodic systems \cite{MonacoPanatiPisanteTeufel2018}, $d=2$ or $3$:  
\renewcommand{\labelenumi}{{\rm(\roman{enumi})}}
\begin{enumerate} 
\item  either there exists a composite Wannier basis whose elements are exponentially localized in space and, correspondingly, the Chern number of the Fermi projection vanishes;
\item  or any composite Wannier basis is such that the expectation value of the squared position operator 
diverges, in the sense that
$$
\sup_{\gamma \in \Gamma_0, 1 \leq a \leq m} \int_{\R^d} \norm{\x - \gamma}^2  | w_{\gamma,a} (\x)|^2 \, \di \x  = + \infty 
$$
where $\Gamma_0 \simeq \Z^d$ is a Bravais lattice. 
\end{enumerate} 

In view of such importance of Wannier bases, it is much desirable to have a corresponding object that can be defined also for non-periodic systems. Indeed, since the early '70, several works have been devoted to the subject. We just mention here the pioneering works of Kohn and Onffroy \cite{KohnOnffroy1973} and Kivelson \cite{Kivelson1982} for non-periodic one-dimensional systems, and the more mathematically oriented works by A. Nenciu and G. Nenciu \cite{NeNe93,NeNe98}. 
Following this stream of ideas, in \cite{MarcelliMoscolariPanati, Moscolari2018} the notion of generalized Wannier basis for any orthogonal projection has been formalized. Without entering into the details of the definition, we say, for example, that an orthogonal projection $P$ acting on $L^2(\R^d)$, $d\geq 1$, admits a \emph{generalized Wannier basis} (GWB) that is exponentially localized if there exist a discrete set $\Lattice \subset \R^d$, a  constant $m_*>0$ and a set $\{\w_{\gamma,a}\}_{\gamma \in \Lattice, 1 \leq a \leq m(\gamma)}\subset L^2(\R^d)$ with $m(\gamma)\leq m_*$ for every $\gamma \in \Lattice$, such that:
	\begin{enumerate}[label=(\roman*),ref=(\roman*)]
		\item $\{\w_{\gamma,a}\}_{\gamma \in \Lattice, 1 \leq a \leq m(\gamma)}$ is an orthonormal basis
		for $\Ran P$;
		\item the functions $\w_{\gamma,a}$ are \emph{uniformly} exponentially localized around the points of $\Lattice$, 
		\ie there exist $\alpha>0$ and $M<\infty$ such that
		\begin{equation*}
		\int_{\R^d} |\w_{\gamma,a}(\x)|^2 \eu^{2 \alpha \|\x-\gamma\|} \,  d\x \leq M \quad \forall\gamma \in \Lattice,\,1 \leq        
		a \leq  m(\gamma).
		\end{equation*}
	\end{enumerate}

In \cite{MarcelliMoscolariPanati} it is shown that GWBs can be used to investigate topological and transport properties of  non-periodic gapped quantum systems. In particular, the Fermi projections that admit an exponentially localized (or just a well-localized \cite{MarcelliMoscolariPanati}) GWB with a  \emph{uniformly discrete} set $\Lattice$, are Chern trivial in the sense that their Chern character is zero. As well known, the Chern character is proportional to the Hall conductivity, thus showing that topological quantum transport and well-localized generalized Wannier bases cannot coexist.

This point of view has been pushed forward and generalized in several directions. Ludewig and Thiang  \cite{LudewigThiang2019} considered systems which are periodic with respect to a suitable non-abelian discrete group, while Bourne and Mesland generalized \cite{MarcelliMoscolariPanati}  to a broader $C^*$-algebraic setting \cite{BourneMesland}. 
Lu and Stubbs enlarged the class of well-localized GWB for which the Chern character of the Fermi projection vanishes
\cite{LuStubbs21I, LuStubbs21II}.  Finally, Ludewig and Thiang realized that \emph{Wannier localizability} is a property of the closed subspaces of a Hilbert space $L^2(X)$ for suitable $X$ (\ie not only of the spectral subspaces of a Schr\"odinger-type operator) corresponding to the triviality of the corresponding orthogonal projection in the $K$-theory of the Roe $C^*$-algebra of $X$ \cite{LudewigThiang2022}, thus paving the way to further developments.
     
Following a parallel and independent line of thought, Prodan constructed, for the spectral subspaces of gapped time-reversal-symmetric Schr\"odinger operators,  a  \emph{radially localized} sort of generalized Wannier basis \cite{Prodan}, which hereafter 
we call Ultra-Generalized Wannier basis (UGWB) to avoid any risk of confusion with the definition of GWB recalled above.
    
In this  paper we first show in Section \ref{sec:ProdanUGWB} that Prodan's construction \cite{Prodan} of an UGWB can be extended to any gapped quantum system, without assuming time-reversal symmetry (Theorem~\ref{TheoremBProdan}), and we prove some additional properties of such bases (Proposition~\ref{prop:InfIDS}). At a first look, the previous Theorem seems to contradict the Localization Dichotomy mentioned above, as the existence of a UGWB is unrestricted, up to minor technical assumptions on the kernel of the corresponding projector, which are typically satisfied by the spectral projections of magnetic Schr\"odinger operators.
However, even though UGWB might be useful to analyze systems with crystalline defects of particular forms as emphasized in \cite{Prodan}, the crucial point is that UGWB are not capable to encode the transport properties of physical systems. Indeed, in Section \ref{Sec:Landau} we consider the explicit example of the Landau operator:  for the orthogonal projection $P_n$ on the $n$-th Landau level, whose Chern number is well-know to be $1$ (up to a sign convention), we explicitly construct an UGWB, elaborating on a result of Raikov and Warzel \cite{RaikovWarzel}. As a consequence, it appears that the existence of an exponentially localized UGWB does not encode relevant information about the transport or topological properties of the physical system.

By contrast, the definition of GWB, while very general and independent of periodicity, still contains relevant topological information as, under the additional assumption that the set $\Lattice$ is uniformly discrete, the existence of a well-localized GWB implies the Chern triviality of the corresponding projection \cite{MarcelliMoscolariPanati}.

\section{Prodan's Ultra-Generalized Wannier Bases}
\label{sec:ProdanUGWB}
The construction of a generalized Wannier basis in dimension $d=1$ is based on the spectral theory of the reduced position operator $\widetilde{X}=P X P$ \cite{Kivelson1982, NeNe98} (see also the more recent generalization to quasi-one dimensional systems \cite{CorneanNenciuNenciu}),  where $P$ is the projection onto an isolated component of the spectrum of a one-dimensional Schr\"odinger operator of the type $-\Delta + V$. In particular, the eigenvalues and eigenvectors of $\widetilde{X}$ are interpreted as points in the position space $\R$ and, respectively, generalized Wannier functions for the spectral projection $P$. While the operator $X$ is an unbounded operator whose spectrum is purely absolutely continuous and covers the entire real line $\R$, the projection of the action of $X$ onto the spectral subspace associated to $P$ creates discrete spectrum \cite{NeNe98}.
Since $P^2=P$, we have that $\varphi \in P L^2(\R)$ is an eigenvector for $\widetilde{X}$ if and only if $\varphi$ is in the kernel of $P \left(X-\lambda\right) P$. Therefore, in the range of $P$, one interprets the eigenvalues of $\widetilde{X}$ as points in the space, since
$$
XP \varphi = \lambda P \varphi + \varphi_{\perp} \,
$$
where $ P \varphi_{\perp} = 0$. This fundamental idea is behind  both the construction of the one-dimensional generalized Wannier basis and the ultra-generalized Wannier basis. Notice that the same argument holds true if we consider $f(X)$ in place of $X$. However, it is necessary for $f$ to be invertible in order to recover a true lattice from the spectrum of the operator $f(X)$.

It is well-known that in $d>1$ the operator $P X_j P$ is not necessarily compact. However, if one considers a suitable function of the position operators, namely $f(X_1,\dots,X_d)=f(\bf{X})$, it is possible to overcome the compactness problem. This is exactly the simple but successful key idea in the paper by Prodan. We notice that more recently there has been another proposal by Stubbs, Lu and Watson to overcome such lack of compactness under further spectral assumptions on the operator $P X_j P$ (namely the uniformity of spectral gaps), see  \cite{StubbsWatsonLu1,StubbsWatsonLu2}. 

In \cite{Prodan} the author considers Schr\"odinger operators of the type $-\Delta + V$, namely only non-magnetic Schr\"odinger operators. The proof in \cite{Prodan} is based on Combes-Thomas estimates and the trace class properties of operators which are of the form $(-\Delta-z)^{-1}f(\mathbf{X})$, and it can be easily generalized to the magnetic case, see Remark~\ref{MagneticProdan}. However, instead of doing so, we take here a slightly different route. We extend Prodan's construction to the case of orthogonal projections that have an integral kernel that decays sufficiently fast. As it is well-known, in dimension $d\leq 3$, spectral projections onto an isolated component of the spectrum of a  \virg{reasonable} Schr\"odinger operator  have an integral kernel that is exponentially localized, see for example \cite{MarcelliMoscolariPanati} and the references therein. Since the argument in \cite[Proposition 2.4]{MarcelliMoscolariPanati} is only sketched, we give in Appendix \ref{sec:CombesThomas} a short proof for the sake of completeness.

\begin{rmk}
\label{MagneticProdan}
 In this remark we briefly explain how the proof in \cite{Prodan} can be generalized to the case of magnetic Schr\"odinger operators. First, the proof of the optimal Combes-Thomas norm estimates given in \cite{Prodan} is based on the results presented in \cite{BaCoHi}, where the magnetic field is already taken into account. Then, it is not difficult to get the optimal Combes-Thomas norm estimates also in the magnetic case. Furthermore, by exploiting the diamagnetic inequality one can show that $(-\Delta_A-z)^{-1}f(\bf{X})$ is compact (where $\Delta_A$ denotes the magnetic Laplacian) whenever $(-\Delta-z)^{-1}f(\bf{X})$ is compact, see for example \cite{AvHeSi78}. By using these two facts, one can generalize the proof \cite{Prodan} to magnetic Schr\"odinger operators.
\end{rmk}

As anticipated, we consider in this paper a general setting. Let us start with a few definitions.

\begin{dfn}[Localization function]
We say that a continuous function \[G\colon[0,+\infty)\to[1,+\infty)\] is a \emph{localization function} if $\lim_{x \to\infty}G(x)=+\infty$ and there exists a constant $C_G>0$ such that
\begin{equation}
\label{eqn:GTriang}
G(\norm{\x-\y})\leq C_G \,G(\norm{\x-\z})G(\norm{\z-\y})\qquad \forall \,\x,\y,\z\in\R^d.
\end{equation}
For $G$ as above, we say that a measurable function $f: \R^d \to \C$ is $G$-localized if the function 
$\x \mapsto G(\norm{\x})f(\x)$, hereafter denoted by $Gf$, is in $L^2(\R^d)$.
\end{dfn}

\begin{dfn}[$G$-localized projection]
	\label{dfn:ELP}
	We say that an orthogonal projection $P$ acting on $L^2(\R^d)$ is \emph{$G$-localized} if $P$ 
	is an integral  operator with a measurable integral kernel $P(\cdot\,,\,\cdot)\colon\R^d\times\R^d \to \C$ and 
	there exists a $G$-localized function $g: \R \to \R$ 
	such that
	\begin{equation}
	\label{DiagLocIntKErnel}
	\left|P(\x,\y)\right|\leq g(\|\x-\y\|) \qquad \forall \, \x,\y \in\R^d.
	\end{equation}
	Furthermore, we say that $P$ is exponentially localized with rate $\beta$ if there exist two constants $C,\beta>0$ such that $g(\|\x-\y\|)\leq C \eu^{-\beta\|\x-\y\|}$  for all  $\x,\y \in\R^d$.
\end{dfn}

Notice that we do not need any regularity, {\it e.g.} continuity, of the integral kernels, what matters is the decay at infinity of the kernels. We are now ready to state our main result.

\begin{thm}[Generalization of \cite{Prodan}]
	\label{TheoremBProdan}
	Let $P$ be an orthogonal projection that is G-localized in the sense of Definition \ref{dfn:ELP}. 
	For $f: \R^d \to \R$ positive and $G$-localized, let $W_f$ be the operator
	\begin{equation}
		\label{DefinitionPrOp}
		W_f:=P f(\mathbf{X})P \,.
	\end{equation}
	Then:
\begin{enumerate} [label=(\roman*),ref=(\roman*)]
		\item $W_f$ is a non-negative Hilbert-Schmidt operator, 
		hence its spectrum consists of positive eigenvalues of finite
		 multiplicity that can possibly accumulate at zero, and zero. 
		 Moreover, every eigenfunction $\psi_\lambda$ corresponding to a positive eigenvalue $\lambda$ is $G$-localized, namely  $G \, \psi_\lambda \in L^2(\R^d)$.
		\item \label{UGWB} Let $\jp{\x}:=\sqrt{1+\|\x\|^2}$. Set $f(\x)=\eu^{-q\jp{\x}}$ for some $q>0$, and assume that $P$ is exponentially localized with rate $\beta>q$. Let $\left\{\lambda_i,\{\psi_{i,j}\}_{ j \leq m_i < \infty}\right\}_{i\in \N}$ be the set of eigenpairs for $W_f$, with eigenvalues $\left\{\lambda_i \right\}_{i \in \N}$ ordered decreasingly. Define 
		\begin{equation}
		r_i:=\sqrt{\left(\frac{\ln(\lambda_i)}{q}\right)^2-1}.
		\end{equation} 
		Then, all the eigenvectors decay exponentially at infinity with a rate $q$ 
		and are radially localized in the sense that  $\exists \, M \in \R$ such that
		\begin{equation}
		\label{eq:Expq}
			\sup_{i,j} \int_{\R^d} \eu^{q |\jp{\x}-\jp{r_i} |} |\psi_{i,j}(\x)|^2 \di\x \leq M.
		\end{equation}
\end{enumerate}
\end{thm}
\begin{proof} The proof of Theorem~\ref{TheoremBProdan} basically follows the argument of \cite{Prodan}, with the exception of step (i) which is considerably simplified in our setting.
\begin{enumerate} [label=(\roman*),ref=(\roman*)]
	\item Since $P$ is a $G$-localized projection, we have that $W_f$ is an integral operator with integral kernel given by
	$$
	W_{f}(\x,\y):= \int_{\R^d}  P(\x,\x') f(\x') P (\x',\y) \di \x'.
	$$ 
	Moreover, $W_f$ is the product of three bounded operators, hence also bounded. Considering the integral kernel of $ P f(\mathbf{X}) $, one has 
	\begin{equation}
	\label{eq:E-qP}
	\begin{aligned}
	&|\left( P f(\mathbf{X}) \right)(\x,\y)|  \leq |\left(G(\|\mathbf{X}\|) P f(\mathbf{X}) \right)(\x,\y)| \\
	&\leq C_G G(\|\x-\y\|) g(\|\x-\y\|) f(\y) G(\|\y\|)  \, ,
	\end{aligned}
	\end{equation}
	where we used that $P$ is $G$-localized, \eqref{eqn:GTriang} and that $G(\|\x\|)\geq 1$ for every $\x \in \R^d$. 
	The estimate \eqref{eq:E-qP} implies that the integral kernel of $ P f(\mathbf{X})$ is in $L^2(\R^d \times \R^d)$, 
	thus $P f(\mathbf{X}) $ is a Hilbert-Schmidt operator. Since Hilbert-Schmidt operators are an ideal, $W_f$ 
	is also a Hilbert-Schmidt operator.
	Then, let $\psi_{\lambda}$ be a normalized eigenvector, $W_f \psi_\lambda=\lambda \psi_\lambda$ 
	for $\psi_\lambda = P \psi_\lambda$. We have that
	\begin{equation}
		\lambda\;=\; \inner{\psi_\lambda}{W_f \psi_\lambda} \;=\; \int_{\R^d} f(\x) |\psi_\lambda(\x)|^2 \di\x >0     \,.
	\end{equation}
	After that, consider the integral kernel of the operator $G(\|\mathbf{X}\|)Pf(\mathbf{X})$. As a by-product of the second inequality of \eqref{eq:E-qP}, we get that $G(\|\mathbf{X}\|)Pf(\mathbf{X})$ is a Hilbert-Schmidt operator.  Thus, we have
	\begin{equation}
	\label{EstimateNorm}
	\begin{aligned}
	&\left(\int_{\R^d} |(G \psi_\lambda)(\x)|^2 \di \x\right)^{\frac{1}{2}} 
	=  \lambda^{-1} \left\| G(\|\mathbf{X}\|) W_f \psi_\lambda\right\|_2  \\
	&= \lambda^{-1} \left\| G(\|\mathbf{X}\|) P 	f(\mathbf{X}) \psi_\lambda \right\|_2 \leq \lambda^{-1} \left\| G(\|\mathbf{X}\|) P 	f(\mathbf{X})\right\|_{HS} \, ,
	\end{aligned}
	\end{equation}
	where $\left\| G(\|\mathbf{X}\|) P 	f(\mathbf{X})\right\|_{HS}$ denotes the Hilbert-Schmidt norm and we have used that $\left\| G(\|\mathbf{X}\|) P 	f(\mathbf{X})\right\| \leq \left\| G(\|\mathbf{X}\|) P 	f(\mathbf{X})\right\|_{HS}$.
	Therefore the eigenfunctions of $W_f$ are such that $G \psi_\lambda \in L^2(\R^d) $.

	\item Here the proof follows the strategy used in \cite{Prodan}. We have that $r_i$ is a sequence of positive numbers which increases monotonically to infinity. Let $\psi_{i,j}$ be an eigenvector for $W_f$ relative to the eigenvalue $\lambda_i$. Then, by choosing $G(x)=\eu^{q \jp{x}}$ in \eqref{EstimateNorm} we get that
	\begin{equation}
	 \int_{\R^d} \eu^{2q(\jp{\x}-\jp{r_i})} |\psi_{i,j}(\x)|^2 \di \x \leq \left\| G(\|\mathbf{X}\|) P 	f(\mathbf{X})\right\|_{HS}^2 \,.
	\end{equation}
	On the other hand
	\begin{equation}
		1\;=\; \lambda_{i}^{-1} \inner{\psi_{i,j}}{W_f\psi_{i,j}} \;= \int_{\R^d} \eu^{q(\jp{r_i}-\jp{\x})} |\psi_{i,j}(\x)|^2 \di\x \, .
	\end{equation}
	Thus, we have
		\begin{equation}
		\begin{aligned}
		\int_{\R^d} \eu^{q|\jp{\x}-\jp{r_i}|} |\psi_{i,j}(\x)|^2 \di \x &= \int_{ \R^d} \chi_{\jp{\x}\geq \jp{r_i}} (\x) \eu^{q(\jp{\x}-\jp{r_i})} |\psi_{i,j}(\x)|^2 \di \x  \\
		&\phantom{=}+ \int_{ \R^d} \chi_{\jp{\x} < \jp{r_i}}(\x) \eu^{q(\jp{r_i}-\jp{\x})} |\psi_{i,j}(\x)|^2 \di \x  \\
		& \leq 	\int_{ \R^d} \eu^{2q(\jp{\x}-\jp{r_i})} |\psi_{i,j}(\x)|^2 \di \x + \int_{ \R^d}  \eu^{q(\jp{r_i}-\jp{\x})} |\psi_{i,j}(\x)|^2 \di \x \\
		&\leq \left\| G(\|\mathbf{X}\|) P 	f(\mathbf{X})\right\|_{HS}^2 +1 .
		\end{aligned}
	\end{equation}
	By defining $M=\left\| G(\|\mathbf{X}\|) P 	f(\mathbf{X})\right\|_{HS}^2 +1$, we obtain the radial localization of the theorem.
\end{enumerate}
\end{proof}

\begin{dfn}
\label{def:UGWB}
We call \emph{Ultra-Generalized Wannier Basis} (UGWB) for the range of $P$ the set of eigenfunctions $\{\psi_{i,j}\}_{1 \leq j \leq m_i < \infty; i\in \N}$ of the operator $P\eu^{-q\jp{\x}}P$ described in Theorem~\ref{TheoremBProdan}.\ref{UGWB}. In particular, the functions $\psi_{i,j}$, $1 \leq j \leq m_i < \infty$, are radially localized around the sphere of radius $r_i$ in the sense of \eqref{eq:Expq}.
\end{dfn}

\medskip
\begin{rmk}
Consider $W_f$ with  $f(\x)=\eu^{-q\jp{\x}}$ like in Theorem \ref{TheoremBProdan}.\ref{UGWB}. Since $f$ depends on the norm of $\x$ it is not possible to uniquely associate to the spectrum of $W_f$ a $d$-dimensional lattice. Nevertheless, it is possible to identify a sequence of concentric $d$-dimensional spheres, around which each ultra-generalized Wannier function is concentrated. Although on the one hand this particular localization shape can be useful for some radially symmetric problems \cite{Prodan}, on the other hand the radial localization clearly breaks the translation symmetry, that is, even in the case of a periodic system the UGWB cannot be built by acting with the translation group on a finite set of functions.
\end{rmk}

\medskip

In \cite{MarcelliMoscolariPanati} the authors showed that if a projection $P$ admits an exponentially localized generalized Wannier basis localized around a uniformly discrete set $\Gamma$, then the Chern character of the projection is zero. 
The proof in \cite{MarcelliMoscolariPanati} relies on two important ingredients: (i) the fact that the generalized Wannier functions are (uniformly) localized around a uniformly discrete set and, (ii) the existence of an upper bound on the number of generalized Wannier functions localized around each point $\gamma \in \Gamma$. As it has been already pointed out in \cite{Prodan}, the construction in Theorem \ref{TheoremBProdan} does not give much information on the structure of the spectrum of $W_f$, in particular it might be that there is no upper bound on the dimension of the eigenspaces associated 
with each $\lambda_i$,  or it might happen that the radii of the annuli are not a uniformly discrete subset of $\R_{+}$. 
In the next proposition we show that the previous situation corresponds to the generic case:  if $P$ has a non-vanishing trace per unit volume, then either the set of radii $\set{r_i} \subset \R_+$ is not uniformly discrete,  or there is no upper bound on the multiplicity of the eigenfunctions of $W_f$, \ie  there is no upper bound on the number of ultra-generalized Wannier functions localized around a certain $d$-dimensional annulus of radius $r_i$.

\begin{prop}
\label{prop:InfIDS}
Let $P$ be an exponentially localized projection with rate $\beta$ in $L^2(\R^d)$, for $d \geq 2$.
Assume that $P$ admits an UGWB in the sense of Definition~\ref{def:UGWB}. Moreover, assume that $\sup_{i} m_i = m_*<+\infty$ and that $\inf_{i,j}|{r_i}-{r_j}|=\delta>0$. Then, $\chi_L P$ is a trace class operator and the trace per unit volume of $P$ is equal to zero, that is
$$
\lim_{L \to \infty} \frac{\Tr(\chi_{\Lambda_L}P)}{|\Lambda_L|} = 0,
$$
where $\chi_{\Lambda_L}$, for $L >0$, denotes the characteristic function of the set $\Lambda_L:=[-L,L]^d$ and $|\Lambda_L|$ is the $d$-dimensional volume of $\Lambda_L$.
\end{prop}
\begin{proof}
First of all note that $\chi_{\Lambda_L}P$ being trace class follows from the exponential localization of the integral kernel of $P$, see for example \cite{MarcelliMoscolariPanati}. Then, we have that
$$
\begin{aligned}
&\Tr(\chi_{\Lambda_L}P)=\sum_{i \in \N} \sum_{j=1}^{m_i} \inner{ \psi_{i,j}}{ \chi_{\Lambda_L} \psi_{i,j}} \\
&= \sum_{i \, \text{s.t.} \, \jp{r_i} \leq \jp{\sqrt{2}L} } \sum_{j=1}^{m_i} \inner{ \psi_{i,j}}{ \chi_{\Lambda_L} \psi_{i,j}} + \sum_{i \, \text{s.t.} \, \jp{r_i} > \jp{\sqrt{2}L} } \sum_{j=1}^{m_i} \inner{ \psi_{i,j}}{ \chi_{\Lambda_L} \psi_{i,j}}=: A + B .
\end{aligned}
$$
Since $\inf_{i,j}|{r_i}-{r_j}|=\delta>0$, we have that the number of radii such that $\jp{r_i} \leq \jp{\sqrt{2}L} $ is bounded by $\sqrt{2}L/\delta$. Thus we get by Cauchy-Schwarz inequality that $|A| \leq (m_* \sqrt{2}L)/\delta$. 

Moreover, assume that $\jp{r_i}>\jp{\sqrt{2}L}$, then we have
$$
\sup_{1\leq j \leq m_i}\|\chi_{\Lambda_L} \psi_{i,j} \|^2 \leq \left(\sup_{\x \in \Lambda_L}  \eu^{- q |\jp{r_i}-\jp{\x}|}\right) \int_{\R^d} \eu^{ q |\jp{\x}-\jp{r_i}|}  |\psi_{i,j}(\x)|^2 \di \x \leq  \eu^{- q |\jp{r_i}-\jp{\sqrt{2}L}|} M.
$$
Therefore we get that
$$
|B| \leq m_* M  \sum_{i \, \text{s.t.} \, \jp{r_i} > \jp{\sqrt{2}L} } \eu^{- q |\jp{r_i}-\jp{\sqrt{2}L}|} \leq \frac{2 m_* C}{q}.
$$
Thus we showed that $\Tr(\chi_{\Lambda_L}P)$ grows at most linearly in $L$, and since $|\Lambda_L|=(2L)^d$, for $d \geq 2$ the thesis follows.  
\end{proof}

\section{UGWB for the Landau operator}	
\label{Sec:Landau}
In the previous section we showed the existence of an ultra-generalized Wannier basis for every orthogonal projection with an exponentially localized integral kernel. In particular, this shows that a UGWB exists for every projection associated to an isolated component of the spectrum of a Schr\"odinger operator, irrespectively of the Chern character of the projection, as detailed in the Appendix. As an explicit example, in the following we construct the UGWB for the projection on any Landau level, thus showing that the existence of a UGWB is insensitive to the topology of the Fermi projection. Since the operator considered in Theorem\ref{TheoremBProdan}.\ref{UGWB} is constructed using a function $f$ that is radially symmetric, $W_f$ reduces in the special setting of the Landau Hamiltonian to a Toeplitz operator, a class of operators extensively studied in the literature \cite{RaikovWarzel}, see also the recent review \cite{Yoshino2017}. 

The Landau Hamiltonian, describing a charged point particle moving under the influence of a constant magnetic field $b>0$ perpendicular to the $xy$ plane, is defined by
\begin{equation}
	\label{LandauHamiltonian}
	H_L:= \half (-\iu \nabla - b{\bf A}_L)^2 \, 
\end{equation}
where ${\bf A}_L$ is the magnetic potential corresponding to a constant magnetic field in the symmetric gauge, that is ${\bf A}_L(\x)=\frac{1}{2}(-x_2,x_1)$.  
$H_L$ is essentially selfadjoint on the dense domain $C^{\infty}_0(\R^2)$ in $L^2(\R^2)$, its spectrum is purely point spectrum given by the eigenvalues
\begin{equation}
	E_n= \frac{b}{2}(2n + 1)\, , \qquad  n \in \N,
\end{equation}
where each $E_n$ is infinitely degenerate. With a little abuse of terminology, the name $n^{th}$ Landau level refers to both the eigenvalue $E_n$ and the corresponding eigenspace. A special orthonormal basis for the $n^{th}$ Landau level  can be written in terms of the Laguerre polynomial. For $\mathbf{x} \in \mathbb{R}^2, n \in \mathbb{N}$, and $k \in \mathbb{Z}_{+}-n:=\{-n,-n+1, \ldots\}$ one defines
$$
\varphi_{n, k}(\mathbf{x}):=\sqrt{\frac{n !}{(k+n) !}}\left[\sqrt{\frac{b}{2}}(x_1+i x_2)\right]^k \mathrm{~L}_n^{(k)}\left(\frac{b\|\mathbf{x}\|^2}{2}\right) \sqrt{\frac{b}{2 \pi}} \eu^{-\frac{b\|\mathbf{x}\|^2}{4}}
$$
where
$$
\mathrm{L}_n^{(\alpha)}(\xi):=\sum_{m=0}^n\left(\begin{array}{c}
n+\alpha \\
n-m
\end{array}\right) \frac{(-\xi)^m}{m !}, \quad \xi \geq 0
$$
are the generalized Laguerre polynomials which are defined using the binomial coefficients $\left(\begin{array}{c}\alpha \\ m\end{array}\right):=\alpha(\alpha-1) \cdot \ldots \cdot(\alpha-m+1) / m !$ if $m \in \mathbb{Z}_{+} \backslash\{0\}$, and $\left(\begin{array}{c}\alpha \\ 0\end{array}\right):=1$ for all $\alpha  \in \R$.

Let $P_n$ be the projection on the $n^{th}$ Landau level. We consider the operator 
$$
W_n:= P_n \eu^{-q\jp{\mathbf{X}}} P_n \, , \qquad q>0.
$$
In this setting Theorem \ref{TheoremBProdan} for $W_f\equiv W_n$, reduces to \cite[ Lemma~3.2 and Lemma~3.3]{RaikovWarzel}, in which the corresponding eigenfunctions and eigenvalues are provided. For completeness we repeat here the explicit construction, which is just a simple computation. Notice that we can choose any $q>0$ in the definition of the operator $W_n$ because each of the integral kernels of the projections $P_n$ has a Gaussian decay.

\begin{prop}
	\label{prop:UGWBLLL}
	The eigenpairs $\left\{\lambda_{n,k},\{\varphi_{n,k}\}_{k \in \mathbb{Z}_{+}-n }\right\}$ provide an ultra-generalized Wannier basis for the projection onto the $n^{th}$ Landau level, and each of the eigenvalues $\lambda_{n,k}$ is explicitly given by 
	\begin{equation}
		\lambda_{n,k} \;=\; \frac{n !}{(k+n) !} \int_0^{\infty}\,  \eu^{-q{\jp{\sqrt{2 \xi / b}}}} \eu^{-\xi} \xi^k \mathrm{~L}_n^{(k)}(\xi)^2 \di \xi, \quad k \in \mathbb{Z}_{+}-n \, .
	\end{equation}
\end{prop}
\begin{proof}
From Theorem \ref{TheoremBProdan} we have that $W_n$ is selfadjoint and bounded. Moreover, it is a standard fact that $\{\varphi_{n,k}\}_{k \in \mathbb{Z}_{+}-n }$ are an orthonormal basis for the range of $P_n$. Furthermore, we can explicitly compute the action of $W_n$ on such orthornomal basis, that is
$$
\begin{aligned}
&C^{-1}_{n,k,k'}\inner{ \varphi_{n,k}}{ W_n \varphi_{n,k'} }\\&=  \int_{\R^2} \frac{b^{\frac{(k+k')}{2}}}{2^{\frac{(k+k')}{2}}} \left[(x_1-i x_2)\right]^k \hspace{-0.1cm} \left[(x_1+i x_2)\right]^{k'} \hspace{-0.1cm}\eu^{-q\jp{\|\x\|}} \mathrm{L}_n^{(k)}\left(\frac{b\|\mathbf{x}\|^2}{2}\right) \mathrm{L}_n^{(k')}\left(\frac{b\|\mathbf{x}\|^2}{2}\right)  \eu^{-\frac{b\|\mathbf{x}\|^2}{2}} \di\x \\
&=\frac{1}{b} \int_{0}^{2\pi}  \eu^{\iu (k'-k)\theta}\di \theta \int_{0}^{\infty}  \, \eu^{-q{\jp{\sqrt{2 \xi / b}}}} \xi^{\frac{k+k'}{2}} \mathrm{~L}_n^{(k)}(\xi) \mathrm{~L}_n^{(k')}(\xi)\eu^{-\xi} \di \xi= \delta_{k,k'} 	\lambda_{n,k}
\end{aligned}
$$
where $C_{n,k,k'}:=\left(\frac{b}{2 \pi}\right) \sqrt{\frac{n !}{(k+n) !}} \sqrt{\frac{n !}{(k'+n) !}}$ and in the second to last equality we have used polar coordinates and the change of variable $\xi=\frac{b\|\x\|^2}{2}$.
\end{proof}

It is interesting to obtain an estimate of the growth of the localization radii $ r_{n,k}$ of the ultra-generalized Wannier basis exhibited above. For simplicity we consider the case of the lowest Landau level.  First, we notice that
\begin{equation}
	\lambda_{0,k} \leq \frac{\eu^{-q}}{k!} \int_{0}^{\infty}  \eu^{-\xi} \xi^k \di \xi = \eu^{-q}\,,
\end{equation}
\begin{equation}
	\lambda_{0,k} \geq \frac{1}{k!} \int_{0}^{\infty} \eu^{-q(1+\frac{2\xi}{b})} \eu^{-\xi} \xi^k \di \xi \;=\; \frac{\eu^{-q}}{k!} \int_{0}^{\infty} \eu^{-\xi(1+\frac{2q}{b})} \xi^k \di \xi \, .
\end{equation}
Thus we get
\begin{equation}
	\lambda_{0,k} \geq \frac{\eu^{-q}}{k!(1+\frac{2q}{b})^{k+1}} \int_{0}^{\infty} \eu^{-s} s^k \di s \;=\; \eu^{-q} (1+\frac{2q}{b})^{-(k+1)} \,.
\end{equation}
Summing up we obtain
\begin{equation}
	1\leq \;  \jp{r_{0,k}} \; \leq 1 + \frac{(k+1)}{q} \ln(1+\frac{2q}{b}) \, .
\end{equation}

\medskip

As it is well-known, every Landau level has a Chern number equal to one (in suitable units) see for example \cite{AvronSeilerSimon1994, DeNittisMoscolariGomi} for explicit computations. Therefore, Proposition~\ref{prop:UGWBLLL} provides an explicit example of a  UGWB for a system that is not time-symmetric and with non trivial topological features. Furthermore, since each of the eigenspaces associated with $\lambda_{n,k}$ is one-dimensional and the integrated density of states of $P_n$ is proportional to the magnetic field $b$, Proposition \ref{prop:InfIDS} implies that $\inf_{i,j}|{r_{n,i}}-{r_{n,j}}|=0$, hence the set of radii is not uniformly discrete.

The strikingly simple structure of the operator $W_n$ described in Proposition \ref{prop:UGWBLLL} is due to the fact that such operator reduces exactly to a Toeplitz operator in the Segal-Bargmann representation \cite{GirvinJach1984} (see also \cite{MoscolariPanati} for a recent review on the subject). Let us briefly show this reduction in the simpler setting of the lowest Landau level. 

Consider the Gaussian measure $\di\mu := N \eu^{-\frac{b}{4}|z|^2} \di z$, with $N$ positive constant, and define the weighted $L^2$-space
 \begin{equation*}
 	L^2(\C,\di \mu) := \left\{g: \C \to \C : \int_{\C}|g(z)|^2 \di\mu < \infty\right\}
 \end{equation*} 
 endowed with the scalar product
 \begin{equation*}
 	\langle f,g \rangle_{SB} := \int_{\C} \overline{f(z)} g(z) \di\mu \, .
 \end{equation*}
Then, the Segal-Bargamm space is defined as follows.
\begin{dfn}[Segal \cite{Segal1963}, Bargmann \cite{Bargmann1961}]
	Let $\Hol(\C)$ be the space of entire functions. The Segal-Bargmann space $SB(\C)$ is defined as 
	\begin{equation*}
		SB(\C):=\left\{ g\in \Hol(\C) : \int_{\C}|g(z)|^2 \di\mu < \infty \right\} \;=\; L^2(\C,\di\mu) \cap \Hol(\C) \;.
	\end{equation*}
\end{dfn}
The unitary operator $U:P_0 L^2(\R^2) \to SB(\C) $ that maps the lowest Landau level onto the Segal-Bargmann space is given by
$$
(U \psi)(z) = f(z) \, 
$$
where $\psi(\x)=f(\x) \varphi_{0,0} (\x)$ via the usual identification of $\R^2$ with the complex plane, $z=x_1+\iu x_2$. 
We denote by $\Pi_0$ the projection in $L^2(\C,\di\mu)$ onto the Segal-Bargmann space $SB(\C)$. Notice that by setting $b=4$, $N=\frac{1}{\pi}$ we recover the standard definition of the Segal-Bargmann space.

In this setting, the operator $W_0$, is a particular restriction of a \emph{Toeplitz operator} \cite{Yoshino2017}.
 \begin{dfn}[Toeplitz operator]
	Let $F$ be a bounded measurable function on the complex plane $\C$, and $M_F$ the multiplication operator in $L^2(\C,\di\mu)$ associated with the function $F$, that is $(M_F g)(z)=F(z)g(z)$. The operator 
	\begin{equation}
		T_F:= \Pi_0 \, M_F
	\end{equation}  
	is called Toeplitz operator associated with the symbol $F$.
\end{dfn}
Therefore, $W_0$ is unitarily equivalent to the restriction to $\Ran \Pi_0$ of the Toeplitz operator associated with the symbol $\ell(z):=\eu^{-q\jp{z}}$:
\begin{equation}
	U W_0 U^*\;=\;\Pi_0 \, T_\ell \, \Pi_0\;\;.
\end{equation}

\textbf{Acknowledgements.} 
M.M. is grateful to H. Cornean for stimulating discussions about Combes-Thomas estimates.
G.P. is grateful to G.C. Thiang for an interesting exchange of ideas on related topics.
The work of M.M. has been supported by a fellowship of the Alexander von Humboldt Foundation. M.M. gratefully acknowledges the support of PNRR Italia
	Domani and Next Generation Eu through the ICSC National Research Centre for High Performance Computing, Big Data and Quantum Computing. G.P. gratefully acknowledges financial support from the National Quantum Science and Technology Institute  (PNRR MUR project PE0000023-NQSTI). The authors gratefully acknowledge the support of the National Group of Mathematical Physics (GNFM-INdAM).

\textbf{Data Availability Statement.}
Data sharing is not applicable to this article as no new data were created or analyzed in this study.

\appendix
\section{Combes-Thomas estimates}
\label{sec:CombesThomas}
Combes-Thomas estimates are ubiquitous in the analysis of Schr\"odinger operators and several proofs can be found in the literature, see for example \cite{CT73,Simon,BaCoHi,CorneanNenciuThermodynamic}. In this appendix we adapt the proofs presented in \cite{CorneanNenciuThermodynamic} for Schr\"odinger operators with smooth potentials and in dimension $d=3$ to our more general setting. We consider magnetic Schr\"odinger operators in $L^2(\R^d)$, with $d=2,3$, namely
\begin{equation}
\label{H}
H=(-\iu\nabla - {\bf A})^2 + V = -\Delta_A + V
\end{equation}
where we assume that the magnetic vector potential ${\bf A}: \R^d \to \R^d$ is in $L^4_{{\rm loc}}(\R^d,\R^d)$ with distributional derivative $\nabla \cdot {\bf A} \in L^2_{{\rm loc}}(\R^d)$ and that $V$ is in $L^2_{\rm uloc}(\R^d)$, which means that $V$ is uniformly locally square-integrable, \ie
\begin{equation}
\label{UniformL2loc}
\sup_{\x \in \R^d} \int\limits_{\|\x-\y\|\leq 1} |V (\y)|^2  \di \y \, < \, \infty.
\end{equation}
From general results on Schr\"odinger operators it follows that $H$ is essentially selfajoint on $C^{\infty}_c(\R^d)$ \cite[Theorem~3]{LeinfelderSimader1981}.
In the following we make use of the notation $-\Delta_A=\sum^d_{i=1}(P_A)_i^2$, where $(P_A)_i:=(-\iu \nabla -{\bf A})_i$, and $\eu^{s{\langle \cdot -\x_0 \rangle}}$ to denote the multiplication operator by the function $\x \mapsto \eu^{s \langle \x -\x_0\rangle}$, $s \in \R$, $\x_0 \in \R^d$. 
\begin{prop}
	\label{prop:gCTEstimates}
	Assume that $z \in D_\eta:= \left\{ z \in \C \; | \; \dist(z,\sigma(H))>\eta>0 \right\}$, $\eta<1$. Then for $i \in \left\{1,\dots d\right\}$ 
	there exists a constant $C$ such that
	\begin{equation}
	\label{gAuxCT2}
	\sup_{ z \in D_\eta } \langle z \rangle^{-1} \left\| \left(\V_A\right)_i \left(H - z\right)^{-1}\right\| \leq \frac{C}{\eta} \, .  
	\end{equation}
\end{prop}
\begin{proof}
	Consider $\lambda \in \R$, for every $\psi \in L^2(\R^d)$, $\|\psi\|=1$, we have
	$$
	\label{gPrelCT}
	\sum_{i=1}^d \left\|\left(\V_A\right)_i \left( -\Delta_A - \iu \lambda  \right)^{-1} \psi\right\|^2 =\textrm{Re}\left(\inner{ \left( -\Delta_A -\iu\lambda  \right)^{-1} \psi }{ \psi} \right)  \leq \frac{C}{|\lambda|} \, .
	$$
	Since $V$ is relatively bounded with respect to $H$, there exists a $\lambda$ such that ( see \cite[Proposition~2.42]{Amrein2009})
	\begin{equation}
	\label{gCTAux}
	\left\|V\left(H-\iu\lambda \right)^{-1}\right\| \leq \frac{1}{2} \, .
	\end{equation}
	By using the resolvent identity 
	$$
	\left(H-z\right)^{-1}=\left(H-\iu\lambda\right)^{-1} + (z-\iu\lambda)\left(H-\iu\lambda\right)^{-1}\left(H-z\right)^{-1}
	$$
	together with \eqref{gCTAux} we get
	\begin{equation}
	\label{gCTAux2}
	\left\|\left(-\Delta_A - \iu \lambda \right) 	\left(H-z\right)^{-1} \right\| \leq 2\left(1 + \frac{|z|+|\lambda|}{\eta} \right) \, .
	\end{equation}
	Therefore, by writing 
	$$
	\left\|\left(\V_A\right)_i 	\left(H-z\right)^{-1}\right\| \leq  \left\| \left(\V_A\right)_i \left( -\Delta_A - \iu \lambda  \right)^{-1}\right\|  \left\| \left( -\Delta_A - \iu \lambda  \right) \left(H-z\right)^{-1}\right\|
	$$
	we obtain the estimate \eqref{gAuxCT2}.
\end{proof}

\begin{prop}[Combes--Thomas estimates]
	\label{prop:gCT}
	Assume that $z \in K$ where $K$ is a compact subset of $D_\eta$. 
	Denote by $\bar{r}=\sup_{z \in K} \jp{z} $. Then there exist a $\delta_0$ and a constant $C$ such that for every $0\leq \delta\leq \delta_0$ we have
	\begin{align}
	\label{gCT}
	&\sup_{z \in K} \sup_{ \x_0 \in \R^d} \left\| \eu^{\pm\frac{\delta}{ \bar{r} } \langle \cdot -\x_0 \rangle  } \left(H-z\right)^{-1} \eu^{\mp\frac{\delta}{\bar{r}} \langle \cdot - \x_0 \rangle  } \right\| \leq \frac{C}{\eta}\, ,\\
	\label{gCT1.1}
	&\sup_{z \in K} \sup_{ \x_0 \in \R^d} \left\{ \bar{r}^{-1} \left\| \left(\V_A\right)_i \eu^{\pm\frac{\delta}{\bar{r}} \langle \cdot -\x_0  \rangle } \left(H-z\right)^{-1} \eu^{\mp\frac{\delta}{\bar{r}} \langle \cdot - \x_0 \rangle  } \right\| \right\} \leq \frac{C}{\eta}\, \, . 
	\end{align}
\end{prop}
\begin{proof}
	For $s \in \R$ the well-known Combes--Thomas rotation gives
	$$
	\eu^{s \langle \cdot -\x_0 \rangle } \left(H-z\right) \eu^{-s \langle \cdot -\x_0 \rangle } = H-z+s\sum_{i=1}^{d}w_i \left(\V_A\right)_i + s W_1 +s^2 W_2 \, 
	$$
	where $w_i$, $W_1$ and $W_2$ are bounded functions uniformly in $\x_0$. Consider now $s=\frac{\delta}{\bar{r}}$. Using \eqref{gAuxCT2} and taking $\delta$ small enough, we obtain
	$$
	\begin{aligned}
	&\sup_{z \in K} \sup_{x_0 \in \R^d}\left\|\left[s\sum_{i=1}^{d}w_i \left(\V_A\right)_i + s W_1 +s^2 W_2\right]\left(H-z\right)^{-1}\right\| \leq \frac{\delta C}{\eta} \left( 1 + \frac{1}{\bar{r}  } + \frac{\delta}{\bar{r}^2  } \right) \leq \frac{1}{2} \, .
	\end{aligned}
	$$
	Therefore we have
	\begin{equation}
	\label{gAuxCT3}
	\begin{aligned}
	&\eu^{s \langle \cdot -\x_0 \rangle } \left(H-z\right)^{-1} \eu^{-s \langle \cdot -\x_0 \rangle }\\&=\left(H-z\right)^{-1} 
	\left\{ 1 + \left[s\sum_{i=1}^{d}w_i \left(\V_A\right)_i + s W_1 +s^2 W_2\right] \left(H-z\right)^{-1} \right\}^{-1} \, .
	\end{aligned}
	\end{equation}
	which implies \eqref{gCT}. Coupling \eqref{gCT} together with \eqref{gAuxCT2} we also obtain the proof of \eqref{gCT1.1}.
\end{proof}

Consider now $\lambda\geq0$ large enough. At the price of enlarging the compact $K$ we can assume that $\lambda \in K$. From \eqref{gAuxCT3} we get for $|s|$ small enough
\begin{equation}
\label{gAuxCT4}
\begin{aligned}
&\left\|\left(-\Delta_A+ \lambda\right) \eu^{s \langle \cdot -\x_0 \rangle } \left(H+\lambda\right)^{-1} \eu^{-s \langle \cdot -\x_0 \rangle } \right\| \\
&= \left\| \left( 1 - V \left(H+ \lambda\right)^{-1} \right)
\left\{ 1 + \left[s\sum_{i=1}^{2}w_i \left(\V_A\right)_i + s W_1 +s^2 W_2\right] \left(H+\lambda\right)^{-1} \right\}^{-1} \right\| \\
&\leq  C_\lambda \, ,
\end{aligned}
\end{equation}
where the constant $C_\lambda$ depends on the $\lambda$ chosen. 
By commuting twice we get
\begin{equation}
\label{eq:AucXT}
\begin{aligned}
\left(-\Delta_A+ \lambda\right) \eu^{s \langle \cdot -\x_0 \rangle }&=  \eu^{s \langle \cdot -\x_0 \rangle } \left(-\Delta_A+  \lambda\right) + \sum_{i=1}^d 2 \left(\V_A\right)_i \left[\left(-\iu \nabla \right)_i, \eu^{s \langle \cdot -\x_0 \rangle } \right]  \\
&\phantom{=}  -  \sum_{i=1}^d  \left[ \left(-\iu \nabla \right)_i, \left[ \left(-\iu \nabla\right)_i, \eu^{s \langle \cdot -\x_0 \rangle } \right]\right]   \, .
\end{aligned}
\end{equation}
Since $\left[\left(-\iu \nabla \right)_i, \eu^{s \langle {\x} -\x_0 \rangle } \right]= - \iu s \partial_i \langle \x -\x_0 \rangle \eu^{s \langle \mathbf{x} - \mathbf{x}_0 \rangle}$, using \eqref{eq:AucXT} together with \eqref{gAuxCT4}, \eqref{gCT1.1}, and setting $s=\frac{\delta}{\bar{r}}$ with $\delta$ small enough we obtain
\begin{equation}
\left\|\eu^{\pm \frac{\delta}{\bar{r}} \langle \cdot -\x_0 \rangle } \left(-\Delta_A+ \lambda\right) \left(H+ \lambda\right)^{-1} \eu^{\mp \frac{\delta}{\bar{r}} \langle \cdot -\x_0 \rangle } \right\| \leq C_{\eta,\lambda}\, ,
\end{equation}  
where $C_{\eta,\lambda}$ is a positive constant that depends only on $\eta$ and $\lambda$. In the same way, we can show that for $c$ small enough we have
\begin{equation}
\label{gCTlambda}
\left\|\eu^{\pm c \sqrt{\lambda}\langle \cdot -\x_0 \rangle } \left(-\Delta_A+ \lambda\right) \left(H+ \lambda\right)^{-1} \eu^{\mp c \sqrt{\lambda} \langle \cdot -\x_0 \rangle } \right\| \leq C_{\lambda}\, ,
\end{equation} 
where the positive constant $C_\lambda$ depends only on $\lambda$.
We now use these norm estimates to get the exponential decay of the integral kernel of the resolvent of $H$.  Consider $\lambda>0$ large enough. Notice that in the following we do not keep track of the $\lambda$-dependence of the constants, while we denote by $C$ any generic positive constant. By using the diamagnetic inequality, see for example \cite{Simon,BroderixHundertmarkLeschke}, we get that
\begin{equation}
\label{FEIntKernelLapl}
\left|\left(-\Delta_A+\lambda\right)^{-1}(\x,\x')\right| \leq \left(-\Delta + \lambda \right)^{-1}(\x,\x').
\end{equation}
In dimensions $d=2$ and $d=3$, the integral kernel of the resolvent of the Laplacian decays exponentially far from the diagonal and has an $L^2$-integrable singularity on the diagonal, for example in $d=2$ we have $$\left(-\Delta + \lambda \right)^{-1}(\x,\x')\leq C \eu^{-\sqrt{\lambda} \|\x-\x'\|} \left(2+ \left|\ln{ \|\x-\x'\|}\right| \right). $$
We are now ready to extract from the $L^2$-norm Combes--Thomas estimate an $L^2$ to $L^\infty$ estimate. Let us see more precisely how it works. From the explicit estimate \eqref{FEIntKernelLapl}, we deduce that there exists a positive constant $c$ such that 
$$
\sup_{\x_0 \in \R^d} \left\|\eu^{\mp c \sqrt{\lambda} \langle \cdot -\x_0 \rangle}\left(-\Delta_A+\lambda\right)^{-1} \eu^{\pm c \sqrt{\lambda} \langle \cdot -\x_0 \rangle }\right\|_{\mathcal{B}(L^2,L^\infty)} <\infty \, .
$$

This, together with the $L^2$ estimate \eqref{gCTlambda}, gives 
\begin{equation*}
\begin{aligned}
&\sup_{\x_0 \in \R^d}\left\|\eu^{\pm c \sqrt{\lambda} \langle{\cdot -\x_0} \rangle }\left(H+\lambda\right)^{-1} \eu^{\mp c \sqrt{\lambda} \langle \cdot - \x_0 \rangle}\right\|_{{\mathcal{B}(L^2,L^\infty)}} \\
&\leq  \sup_{\x_0 \in \R^d}\left\|\eu^{\pm c \sqrt{\lambda}\langle{\cdot -\x_0}\rangle}\left(-\Delta_A+\lambda\right)^{-1} \eu^{\mp c \sqrt{\lambda} \langle{\cdot -\x_0}\rangle }\right\|_{{{\mathcal{B}(L^2,L^\infty)}}} \\
&\phantom {\leq \sup_{x_0 \in \R^d}} \cdot \sup_{\x_0 \in \R^d}\left\| \eu^{\pm c \sqrt{\lambda} \langle \cdot - \x_0 \rangle }\left( -\Delta_A + \lambda\right) \left(H-\lambda\right)^{-1} \eu^{\mp c  \sqrt{\lambda} \langle \cdot - \x_0 \rangle} \right\|_{\mathcal{B}(L^2,L^2)} \, ,
\end{aligned} 
\end{equation*}
hence the operator $\eu^{\pm c \sqrt{\lambda} \langle{\cdot -\x_0}\rangle}\left(H+\lambda\right)^{-1} \eu^{\mp c \sqrt{\lambda} \langle \cdot - \x_0 \rangle}$ is bounded from $L^2$ to $L^\infty$ and it is a Carleman integral operator.

From \eqref{FEIntKernelLapl} we have that the measurable integral kernel $\left(-\Delta_A+\lambda\right)^{-1}(\cdot,\cdot)$ is bounded outside the diagonal, moreover, without loss of generality we consider that \eqref{FEIntKernelLapl} is valid pointwise for every (and not only for almost every) $\x,\x' \in \R^2$ (we choose a representative for the integral kernel  that is continuous outside of the diagonal \cite{Simon,BroderixHundertmarkLeschke}). Then, we have that
$$
\begin{aligned}
&\eu^{ -c \sqrt{\lambda}  \langle{\cdot -\x_0}\rangle}\left(H+\lambda\right)^{-1} \eu^{c \sqrt{\lambda} \langle \cdot - \x_0 \rangle}  
=: \left(\eu^{-c \sqrt{\lambda}  \langle{\cdot -\x_0}\rangle} \left( -\Delta_A + \lambda\right)^{-1}      \eu^{ c \sqrt{\lambda}  \langle{\cdot -\x_0}\rangle}\right) B_{\x_0}
\end{aligned}
$$
where we have set $B_{\x_0}:=\left(\eu^{ -c \sqrt{\lambda}  \langle{\cdot -\x_0}\rangle} \left( -\Delta_A + \lambda\right) \left(H-\lambda\right)^{-1} \eu^{c \sqrt{\lambda}  \langle{\cdot -\x_0}\rangle}\right)$. Consider now, for every $\psi \in L^2(\R^d)$ the map
$$
F_{\x',\x_0}(\psi):= \int_{\R^d}  \eu^{ -c \sqrt{\lambda}  \langle{\x' -\x_0}\rangle}\left( -\Delta_A + \lambda\right)^{-1}(\x',\x)      \eu^{ c \sqrt{\lambda}  \langle{\x -\x_0}\rangle} \left(B_{\x_0}\psi\right)(\x) \di\x .
$$
The map $F_{\x',\x_0}$ defines a bounded linear functional on $L^2$ and its norm is independent on $\x',\x_0$. Indeed
$$
\begin{aligned}
|F_{\x',\x_0}(\psi)| 
& \leq \sup_{ \x_0 \in \R^d} \| \eu^{ -c \sqrt{\lambda}  \langle{\cdot -\x_0}\rangle} \left( -\Delta_A + \lambda\right)^{-1}  \eu^{ c \sqrt{\lambda}  \langle{\cdot-\x_0}\rangle} \|_{\mathcal{B}(L^2,L^\infty)} \|B_{\x_0}\|_{\mathcal{B}(L^2,L^2)} \|\psi\|_2 \, .
\end{aligned}
$$
 Since from \eqref{gCTlambda} we know that the norm of $B_{\x_0}$ does not depend on $\x_0$, we get that $F_{\x',\x_0}$ defines a bounded linear functional on $L^2(\R^d)$ whose norm is independent on $\x_0,\x'$. From Riesz representation theorem we get that there exists a function $f_{\x',\x_0}$ in $L^2(\R^d)$ such that
$$
F_{\x',\x_0}(\psi)=\int_{\R^d} \overline{f_{\x',\x_0}(\x)} \psi(\x) \di \x.
$$
$F_{\x',\x_0}(\psi)$ can be rewritten as 
$$
F_{\x',\x_0}(\psi)=\left(\eu^{ -c \sqrt{\lambda}  \langle{\cdot -\x_0}\rangle}\left(H+\lambda\right)^{-1} \eu^{c \sqrt{\lambda} \langle \cdot - \x_0 \rangle} \psi \right)(\x')
$$
which implies that
$$
\sup_{\x_0,\x' \in \R^d} \| f_{\x',\x_0}(\cdot)\|_{2} = \sup_{\x_0,\x' \in \R^2} \left\| \left(\eu^{ -c \sqrt{\lambda}  \langle{\x' -\x_0}\rangle}\left(H+\lambda\right)^{-1}(\x',\cdot) \eu^{c \sqrt{\lambda} \langle \cdot - \x_0 \rangle}  \right) \right\|_{2} < \infty \, .
$$
By taking $\x'=\x_0$ (namely $F_{\x_0,\x_0}$) and exploiting the selfadjointness of the operators we also get
\begin{equation}
\label{eq:aux1}
\sup_{\x_0 \in \R^d} \left\| \left(H+\lambda\right)^{-1}(\x_0,\cdot) \eu^{c \sqrt{\lambda} \langle \cdot - \x_0 \rangle}   \right\|_{2} =  \sup_{\x_0 \in \R^d} \left\|  \eu^{c \sqrt{\lambda} \langle \cdot - \x_0 \rangle}\left(H+\lambda\right)^{-1}(\cdot,\x_0)   \right\|_{2}< \infty \, .
\end{equation}
Consider now the integral kernel of $\left(H+\lambda\right)^{-2}$, which is a priori defined using the integral kernel of the resolvent. By using \eqref{eq:aux1}, the Cauchy-Schwarz inequality and the triangle inequality, we get that
\begin{equation}
\begin{aligned}
&\sup_{\x, \x' \in \R^d} \left| \eu^{ c \sqrt{\lambda} \|\x-\x'\|} \left(H+\lambda\right)^{-2}(\x,\x') \right| \\
& \leq (\eu^{2c\sqrt{\lambda}}) \sup_{\x,\x'\in \R^d} \int_{ \R^d}  \eu^{ c \sqrt{\lambda} \langle \y-\x \rangle} \left| \left(H+\lambda\right)^{-1}(\x,\y ) \right| \left| \left(H+\lambda\right)^{-1}(\y,\x' ) \right| \eu^{c \sqrt{\lambda} \langle \y- \x' \rangle} \di \y  \\
&  \leq  (\eu^{2c\sqrt{\lambda}}) \sup_{\x \in \R^d} \left\|\left(H+\lambda\right)^{-1}(\x,\,\cdot ) \eu^{ c \sqrt{\lambda} \langle \cdot - \x \rangle}\right\|_{2} \sup_{\x' \in \R^d} \left\| \eu^{ c \sqrt{\lambda} \langle \cdot - \x' \rangle}\left(H+\lambda\right)^{-1}(\cdot\,,\x' )\right\|_{2} \\
&\leq (\eu^{2c\sqrt{\lambda}}) C \, .
\end{aligned}
\end{equation}
Therefore we have obtained that the second power of the resolvent is pointwise exponentially decaying, that is
\begin{equation}
\label{gFEAux8}
\left| \left(H+\lambda\right)^{-2}(\x,\x') \right| \leq  C \eu^{- c \sqrt{\lambda} \|\x-\x'\|} \,.
\end{equation}

Let us analyze how estimate \eqref{gFEAux8} propagates in the resolvent set. Since we are interested in proving the exponential decay of the integral kernel of the projection onto an isolated component of the spectrum, we assume that $z \in K$, with $K$ compact subset of $D_\eta$, as defined in Proposition \ref{prop:gCTEstimates} and \ref{prop:gCT}.
From \eqref{eq:aux1}  we get, for $\delta< c \sqrt{\lambda} \bar{r}$ small enough and for every $\varphi \in L^2(\R^d)$  
$$
\begin{aligned}
& \sup_{ \x_0 \in \R^d} \sup_{ \x \in \R^d} \left| \int_{ \R^d}  \, \eu^{-\frac{\delta}{\bar{r}} \langle \x -\x_0 \rangle } \left(H+\lambda\right)^{-1} (\x,\x') \eu^{\frac{\delta}{\bar{r}} \langle \x' -\x_0 \rangle } \varphi(\x') \di\x' \right| \\
& \leq \eu^{3 \frac{\delta}{\bar{r}}}\sup_{ \x_0 \in \R^d} \sup_{ \x \in \R^d}  \int_{ \R^d}  \left| \left(H+\lambda\right)^{-1} (\x,\x') \right|  \eu^{ c \sqrt{\lambda} \langle \x' -\x \rangle } |\varphi(\x')| \di\x' \\
&\leq \eu^{3 \frac{\delta}{\bar{r}}} \sup_{ \x_0 \in \R^d}  \sup_{ \x \in \R^d}  \left\|  \left(H+\lambda\right)^{-1} (\x,\,\cdot)   \eu^{ c \sqrt{\lambda} \langle \cdot -\x \rangle } \right\|_{2} \left\|\varphi\right\|_{2}\, .
\end{aligned}
$$
Hence
$$
\begin{aligned}
&\sup_{ \x_0 \in \R^d} \left\| \eu^{-\frac{\delta}{\bar{r}} \langle \cdot -\x_0 \rangle } \left(H+\lambda\right)^{-1}  \eu^{\frac{\delta}{\bar{r}} \langle \cdot -\x_0 \rangle } \right\|_{\mathcal{B}(L^2,L^\infty)} \leq  \eu^{3 \frac{\delta}{\bar{r}}}C  \, .
\end{aligned}
$$
This, together with the $L^2$ bound \eqref{gCT} and the resolvent identity, implies that
\begin{equation}
\label{gzLinfty}
\begin{aligned}
&\sup_{ \x_0 \in \R^d} \left\| \eu^{-\frac{\delta}{\bar{r}} \langle \cdot -\x_0 \rangle } \left(H-z\right)^{-1}  e^{\frac{\delta}{\bar{r}} \langle \cdot -\x_0 \rangle } \right\|_{\mathcal{B}(L^2,L^\infty)}\\
& \leq \sup_{ \x_0 \in \R^d} \left\| \eu^{-\frac{\delta}{\bar{r}} \langle \cdot -\x_0 \rangle } \left(H+\lambda\right)^{-1}  \eu^{\frac{\delta}{\bar{r}} \langle \cdot -\x_0 \rangle } \right\|_{\mathcal{B}(L^2,L^\infty)} \\
&\phantom{\leq} +   \left( |z| +|\lambda|\right) \sup_{ \x_0 \in \R^d}\left\|\eu^{-\frac{\delta}{\bar{r}} \langle \cdot -\x_0 \rangle } \left(H+\lambda\right)^{-1} \eu^{\frac{\delta}{\bar{r}} \langle \cdot -\x_0 \rangle } \right\|_{\mathcal{B}(L^2,L^\infty)} \\
&\phantom{\leq +} \cdot \sup_{ \x_0 \in \R^d} \left\| \eu^{-\frac{\delta}{\bar{r}} \langle \cdot -\x_0 \rangle } \left(H-z\right)^{-1} \eu^{\frac{\delta}{\bar{r}} \langle \cdot -\x_0 \rangle } \right\|_{\mathcal{B}(L^2,L^2)} \\
& \leq C  \eu^{3 \frac{\delta}{\bar{r}}} \frac{\bar{r}}{\eta} \, ,
\end{aligned}
\end{equation}
which shows that $\eu^{-\frac{\delta}{\bar{r}} \langle \cdot -\x_0 \rangle } \left(H-z\right)^{-1}  e^{\frac{\delta}{\bar{r}} \langle \cdot -\x_0 \rangle } $ is also a Carleman operator.
Consider now
$$
\eu^{-\frac{\delta}{\bar{r}} \langle \cdot -\x_0 \rangle }\left(H-z\right)^{-1}\eu^{\frac{\delta}{\bar{r}} \langle \cdot -\x_0 \rangle } = \eu^{-\frac{\delta}{\bar{r}} \langle \cdot -\x_0 \rangle }\left( -\Delta_A + \lambda\right)^{-1}\eu^{\frac{\delta}{\bar{r}} \langle \cdot -\x_0 \rangle } B'_{\x_0},
$$
where we have defined
$$
\begin{aligned}
B'_{\x_0}:=&\left(\eu^{-\frac{\delta}{\bar{r}} \langle \cdot -\x_0 \rangle }\left( -\Delta_A + \lambda\right)\left(H+\lambda\right)^{-1}\eu^{\frac{\delta}{\bar{r}} \langle \cdot -\x_0 \rangle }\right)  \\
&+\left( (z+\lambda)\eu^{-\frac{\delta}{\bar{r}} \langle \cdot -\x_0 \rangle }\left( -\Delta_A + \lambda\right)\left(H+\lambda\right)^{-1}\eu^{\frac{\delta}{\bar{r}} \langle \cdot -\x_0 \rangle } \eu^{-\frac{\delta}{\bar{r}} \langle \cdot -\x_0 \rangle }\left(H-z\right)^{-1}\eu^{\frac{\delta}{\bar{r}} \langle \cdot -\x_0 \rangle } \right).
\end{aligned}
$$
We can repeat the same strategy as before by defining a new linear functional $F'_{\x',\x_0}: L^2(\R^d) \to \C$
$$
F'_{\x',\x_0}(\psi):= \int_{\R^d} \eu^{ -\frac{\delta}{\bar{r}}  \langle{\x' -\x_0}\rangle}\left( -\Delta_A + \lambda\right)^{-1}(\x',\x)      \eu^{ \frac{\delta}{\bar{r}} \langle{\x -\x_0}\rangle} \left(B'_{\x_0}\psi\right)(\x) \di\x .
$$
Thus, we obtain
\begin{equation}
\label{eq:Resolvent2}
\begin{aligned}
&\sup_{ \x,\x' \in \R^d} \left|\eu^{\frac{\delta}{\bar{r}} \|\x-\x'\|}  \left(H-z\right)^{-2}(\x,\x') \right|\\
& \leq (\eu^{2\frac{\delta}{\bar{r}}}) \sup_{\x,\x'\in \R^d} \int_{ \R^d} \eu^{ \frac{\delta}{\bar{r}} \langle \y- \x \rangle} \left| \left(H-z\right)^{-1}(\x,\y ) \right| \left| \left(H-z\right)^{-1}(\y,\x' ) \right| \eu^{\frac{\delta}{\bar{r}}\langle \y- \x' \rangle} \di \y \\
&  \leq (\eu^{2\frac{\delta}{\bar{r}}}) \sup_{\x \in \R^d} \left\|\left(H-z\right)^{-1}(\x,\cdot ) \eu^{ \frac{\delta}{\bar{r}} \langle \cdot - \x \rangle}\right\|_{2} \sup_{\x' \in \R^d} \left\| \eu^{ \frac{\delta}{\bar{r}} \langle \cdot - \x' \rangle}\left(H-z\right)^{-1}(\cdot\,,\x' )\right\|_{2} \\
&\leq    \frac{C_{\delta, \bar{r}}}{{\eta^2}} \, ,
\end{aligned}
\end{equation}
where $C_{\delta, \bar{r}}$ is a finite constant that depends on $\delta$ and $\bar{r}$ (and $\lambda$).

Assume now that $H$ has an isolated component of the spectrum $\sigma_0$, so that we can find a countour $\mathcal{C} \subset K \subset \rho(H)$ encircling $\sigma_0$. The projection $P$ onto $\sigma_0$ can be written using the Riesz formula together with integration by parts as
\begin{equation*}
P = - \frac{\iu}{2 \pi} \oint_{\mathcal{C}} z \left(H-z\right)^{-2} \di z \, ,
\end{equation*}
which together with \eqref{eq:Resolvent2} implies that $P$ is an exponentially localized projection in the sense of Definition \ref{dfn:ELP}, that is
\begin{equation*}
\sup_{ \x,\x' \in \R^d} \left|\eu^{\frac{\delta}{\bar{r}} \|\x-\x'\|} P(\x,\x') \right| \leq C\, .
\end{equation*}


\vfill

{
	
	\begin{tabular}{rl}
		(M. Moscolari) & \textsc{Dipartimento di Matematica, Politecnico di Milano}\\ &Piazza Leonardo da Vinci 32, 20133, Milano, Italy \\
		&  \textsl{E-mail address}: \href{mailto:massimo.moscolari@polimi.it}{\texttt{massimo.moscolari@polimi.it}} \\
		\\
		(G. Panati) & \textsc{Dipartimento di Matematica,  Sapienza Università di Roma}\\
		& Piazzale Aldo Moro 2, 00185 Roma, Italy \\
		&  \textsl{E-mail address}:
		\href{mailto:panati@mat.uniroma1.it}{\texttt{panati@mat.uniroma1.it}} \\
	\end{tabular}
	
}

\end{document}